\documentclass[conference]{IEEEtran}

\usepackage{amsmath,amssymb,cite}
\usepackage{colordvi}
\usepackage[colorlinks=true]{hyperref} 
\usepackage{color}
\usepackage{paralist}
\definecolor{rltblue}{rgb}{0,0,0.75}
\usepackage{graphicx}
\newtheorem{lemma}{\textbf{Lemma}}
\newtheorem{theorem}[lemma]{\textbf{Theorem}}
\newtheorem{proposition}[lemma]{Proposition}

\newtheorem{corollary}[lemma]{Corollary}

\newtheorem{remark}{Remark}
\DeclareMathOperator{\ord}{ord}

\begin{document}

\title{ The Permutation Groups and the Equivalence of Cyclic and
  Quasi-Cyclic Codes}
\hyphenation{op-tical net-works semi-conduc-tor}
\author{ Kenza Guenda
\thanks{K. Guenda is with the Faculty of Mathematics USTHB, University
  of Science and Technology of Algiers, Algeria  e-mail: kguenda@gmail.com}}
  
\maketitle
\begin{abstract}
We give the class of finite groups which arise as the permutation
groups of cyclic codes over finite fields.  Furthermore,  we extend the results of Brand and Huffman et al.  and we find the properties of the set  of permutations by
which two cyclic codes of length $p^r$ can be equivalent. We also find
the set of permutations by which two quasi-cyclic codes can be equivalent.
\end{abstract}
\begin{IEEEkeywords}
Permutation group, equivalency of codes, cyclic code, quasi-cyclic code, doubly transitive groups.
\end{IEEEkeywords}
\IEEEpeerreviewmaketitle
\section{Introduction}
A class of cyclic objects on $n$ elements is a class of combinatorial
objects on $n$ elements, where automorphisms of objects in the class
and isomorphisms between objects in the same class are permutations of $S_n$,
and where permutation group contains a complete cycle. Such classes include circulant graphs, circulant
digraphs, cyclic designs and cyclic codes. N. Brand~\cite{brand}
characterised the set $H(P)$ of permutations by which two combinatorial cyclic objects on $p^r$ elements are equivalent. By using these results Huffman et al.~\cite{vanessa} gave explicitly this
set in the case $n=p^2$ and construct algorithms to find the
equivalency between cyclic objects and extended cyclic objects.
In this paper we also give explicitly the set $H(P)$ for
codes of length $p^r$. We simplify the algorithms of Huffman et al. by
proving some results on the order of some subgroups of the
permutation group $Per(C)$.

 It is well known that we can construct from the cyclic codes many optimal
codes with permutation groups sharing many
properties~\cite{conway} or \cite{macwilliams}. With this motivation and also since the set $H(P)$ depends essentially on
the structure of the group $Per(C)$,
we give the class of finite groups which arise as the
permutation groups of cyclic codes.  Note that the
permutation groups of cyclic codes are known only for few families,
such as the Reed--Solomon codes, the Reed--Muller codes and some $BCH$
codes~\cite{berger91,berger96b,berger99}. Recently,
R. Bienert and B. Klopsch~\cite{rolf} studied the permutation group
of cyclic codes in the binary case.
We generalize a result
of~\cite{rolf} concerning the doubly transitive permutation groups with socle $PSL(d,q)$
to the non-binary cases. Furthermore, we prove that if the length
is a composed or a prime power number, then $Per(C)$ is imprimitive
or doubly transitive. Hence, we use the classification of the doubly
primitive groups which contains a complete cycle, given by
J. P. McSorley~\cite{sorley,mortimer} and our previous results to give the
permutation groups in the doubly transitive cases. Further, we consider
the permutation groups of cyclic $MDS$ codes, and that by building on
the results of~\cite{berger93}. Finally, we
consider the quasi-cyclic codes.  These codes are interesting;  they are used in many powerful cryptosystems~\cite{otmani}.
We characterize the set $H'(P)$ of permutations by which
two quasi-cyclic codes can be equivalent. We find some of its properties. But, we did not prove that $H'(P)$ is a group. Even though, by using the
software GAP, we find on several examples that $H'(P)$ is an imprimitive
group. Hence we conjecture that $H'(P)$ is a group. Under this
hypothesis, we prove that $H'(P)$ is an imprimitive group, or the alternating group
$Alt(n)$ or the symmetric group $S_n$. The last situation implies that
the code is trivial. 
 
This paper is organised as follow. In second section we deal with
the permutation groups of cyclic codes. In the third section we
consider the equivalency problem for the cyclic codes. We simplify and
generalize some results of Huffman et al.~\cite{vanessa} from the
length $p^2$ to the length
$p^r$. We
characterize the structure of $H(P)$ and in some cases we give exactly
the set $H(P)$ or prove that it is a group. Finally in the last
section we consider the equivalency problem for the quasi-cyclic codes. 
\section{Permutation Groups of Cyclic Codes }
Let $C$ be a linear code of length $n$ over a finite field $\mathbb{F}_q$, and $\sigma$ a permutation of $S_n$. To the code $C$ we associate a linear code $\sigma(C)$ defined by: 
$$\sigma(C)=\{ (x_{\sigma ^{-1}(1)}, \ldots, x_{\sigma ^{-1}(n) } )\, |\,
  (x_1, \ldots x_n)\in C \}.$$
We say that the codes $C$ and $C'$ are equivalent by permutation
  if there exists a permutation $\sigma \in S_n$ such that $C'=
  \sigma(C)$. 
The permutation group of $C$ is the subgroup of $S_n$ given by:
$$ Per(C)=\{\sigma \in S_n \,|\, \sigma(C)=C\}.$$
 
We recall that a linear code $C$ over $\mathbb{F}_q$ is 
cyclic if it verifies $T \in Per(C)$, where $T$ is a complete cycle of length $n$. 

The following elementary Lemma is a folklore, well-known in the area of group theory.
\begin{lemma}
\label{lem:bur2}
Let $C$ be a cyclic code of  length $n$. Then its permutation group $
Per(C)$ is a transitive group.
\end{lemma}
\begin{proof}
The group $Per(C)$ operates on the set $\{1,\ldots,n\}$ and contains the permutation
shift $T$. For $i,j \in \{1,\ldots,n\}$, we have $T^{j-i}(i)=j $;
hence $Per(C)$ is transitive.
\end{proof}

A transitive group is either primitive or imprimitive. An interesting
class of primitive group is the class of the doubly transitive
groups. A doubly transitive permutation group $G$ has a unique minimal
normal subgroup $N$, which is either regular and elementary abelian or
simple and primitive and has $C_G(N)=1$~\cite[ p. 202]{burnside}. All simple groups which can occur as minimal normal subgroup
of a doubly transitive group are known~\cite{cameron}. This result is
due to the classification of finite simple
groups. By using this classification J. P. McSorley~\cite{sorley} gave the following result. 
 \begin{lemma}
\label{lem:sorley}
 A group $G$ of degree $n$ which is doubly transitive and contains a complete
cycle, has as socle $N$, verifies
$N \leq G \leq Aut(N)$ and is equal to one of the cases in the
following Table.
\end{lemma}
\small{
\label{table. x}
\begin{tabular}{|c|c|c|}
\hline
 $G$ & $n$ & $N$  \\
\hline
$AGL(1,p) $ & $p$ &$C_p$ \\
$S_4 $ & $4$ &$C_2 \times C_2$ \\
$S_n , n\geq 5$ & $n$ &$Alt(n)$ \\
$Alt(n),n \text{ odd and } \geq 5$ & $n$ &$Alt(n)$ \\
 $PGL(d,t) \leq G \leq P\Gamma L(d,t)$ & $(t^d-1)/t-1$&
 $PSL(d,t)$ \\ 
$(d,t)\neq  (2,2),(2,3),(2,4) $&&\\
$PSL(2,11)$& 11 &$PSL(2,11)$\\
$M_{11}$&11&$M_{11}$\\
$M_{23}$&23&$M_{23}$\\
\hline
\end{tabular}
\normalsize{}\\

\begin{remark}
\label{rem:iso}
The projective semi-linear group $P\Gamma L(d,t)$ is
the semi-direct product of the projective linear group $PGL(d,t)$
with the automorphism group $Z=Gal(\mathbb{F}_t/\mathbb{F}_{p'})$ of
finite field $\mathbb{F}_t$, where
$t=p'^s,p'$ a prime, i.e. $$ P\Gamma L(d,t)= PGL(d,t)\rtimes Z.$$ The
order of these groups are $$|PGL(d,t)|=(d,t-1)|PSL(d,t)|, |Z|=s$$ and $ P\Gamma
L(d,t)=s|PGL(d,t)|$. Hence if $t$ is a prime we have $ P\Gamma L(d,t)= PGL(d,t)$.   \\
 The zero
code, the entire space, the repetition code and its dual
are called  elementary codes. The permutation group of these codes is $S_n$~\cite[p. 1410]{hand}. Furthermore, there is no cyclic
codes with permutation group equals to $Alt(n)$.  
\end{remark}
The following Lemma is a generalization of a part of~\cite[Theorem E]{rolf} to the non binary cases.
\begin{lemma}
\label{lem:proj}
Let $C$ be a non elementary cyclic code of length $n=\frac{t^d-1}{t-1}$ over a finite field of
characteristic $p'$, $d \geq 2,$ and $t=p^s$ be the power of a prime $p$, 
if the  group $Per(C)$ verifies:
$$  PGL(d,t) \leq Per(C) \leq P\Gamma L(d,t),$$
 then we have, $$d \geq 3, t=p'^s \text{ and }Per(C)=P\Gamma L(d,t).$$
\end{lemma}
\begin{proof}
For $d=2$. As the group $PGL(2,t)$ acts 3-transitively on
$1$-dimensional projective space $\mathbb{P}^1(\mathbb{F}_t)$, we can
deduce from~\cite[Table 1 and Lemma 2]{mortimer}, that the underlying code is
elementary, hence $Per(C)=S_n$, which is a contradiction.
 Hence $d\geq 3,$ and similarly, we deduce from~\cite[Table 1 and
   Lemma 2]{mortimer} that $t=p'^s$. Let $V$ denotes the permutation
 module of $ \mathbb{F}_{p'}$, associated to the natural action of
 $PGL(d,t)$ on $(d-1)$ dimensional projective space
 $\mathbb{P}^{d-1}( \mathbb{F}_t)$. Let $U_1$ be a $PGL(d,t)$-submodule
 of $V$. Hence $U_1$ is $P\Gamma L(d,t)$-invariant. That is because,
 for $\sigma$ a generator of the cyclic group $P\Gamma L(d,t)/PGL(d,t)
 \simeq Gal( \mathbb{F}_t/ \mathbb{F}_{p'})$. Then $U_2=U_1^{\sigma},$
 regarded as $PGL(d,t)$-module, is simply a twist of $U_1$. Let
 $\overline{\mathbb{F}}_{p'}$ be the algebraic closure of $\mathbb{F}_{p'}$,
 we conclude that the composition factors of the
 $\overline{\mathbb{F}}_{p'}PGL(d,t)$-modules
 $\overline{U_1}=\overline{\mathbb{F}}_{p'}\otimes U_1$ and
 $\overline{U_2}=\overline{\mathbb{F}}_{p'}\otimes U_2$ are the same. The
 submodules of the $\overline{\mathbb{F}}_{p'}PGL(d,t)$-module
 $\overline{V}=\overline{\mathbb{F}}_{p'}\otimes V$ are uniquely
 determined by their composition factors~\cite{bardoe}. Hence we have     
$\overline{U}_1=\overline{U}_2$; this implies $U_1=U_2$. This gives
 $Per(C)=P\Gamma L(d,t)$.  
\end{proof}

We recall that the group $AG(n)=\{\tau _{a,b} : a \neq 0,
(a,n)=1, b \in \mathbb{Z}/n\mathbb{Z}\} $ is the group of the affine transformations
defined as follow
  \begin{equation}
\begin{split}
\tau _{a,b} : Z_n & \longrightarrow  Z_n \\
          x &\longmapsto (ax +b )\mod n.
\end{split}
\end{equation}
 When $n=p$ a prime number, we have $AGL(1,p)=AG(p)$, and is called
 the affine group. 
\begin{proposition}
\label{lem:bur}
Let $C$ be a non elementary cyclic code of length $p$ over
$\mathbb{F}_q$, with $q=p'^{s}$. Hence $Per(C)$ is a primitive
group, and  one of the following holds:
\begin{enumerate}
\item $Per(C)=PSL(2,11)$ of degree $11$, $q=3$ ( $C$ is the $[11,6,5]_3$ ternary Golay code
  or its dual), 
\item $Per(C)=M_{23}$ of degree $23$, $q=2$ ($C$ is the binary Golay
  code $[23,12,7]$, or its dual),
\item $ Per(C)=P\Gamma L (d,t)$ of degree $p=(t^d-1)/(t-1)$, where $t$ is a
  power of  the prime $p'$.
\item $C_p \leq Per(C) \leq AGL(1,p)$, with
  a normal Sylow subgroup $N$ of order $p$; and such that $p \geq 5$.
\end{enumerate}
\end{proposition}
\begin{proof}
A transitive permutation group of prime degree is a primitive group~\cite[p. 195]{robinson}. From a consequence of a result of
Burnside~\cite{burnside}, we have that a transitive group of prime
degree is either a subgroup of the affine group  $AGL(1,p)$ or a doubly transitive group.

In the doubly-transitive cases, from Lemma~\ref{lem:sorley},
 Remark~\ref{rem:iso} and the Lemma~\ref{lem:proj} we have that these groups are $S_p$, $M_{11}$,
 $PSL(2,11)$ of degree 11, $M_{23}$ of degree 23, or $ P\Gamma L(d,t)$ of degree $p=(t^d-1)/(t-1)$ over $\mathbb{F}_q$, with $q=p'^s$
 and $t$ is a prime power of $p'$. From Remark~\ref{rem:iso}, the group $S_p$
 corresponds to an elementary code. From~\cite[Table 1, Lemma 2 and (J)]{mortimer},
  $M_{11}$ rules out, $PSL(2,11)$ is the permutation group of a
  ternary code and $M_{23}$ is the permutation group of a binary
  code. From~\cite{hand} p. 1410, these codes are unique, namely
  $PSL(2,11)$ corresponds to the $[11,6,5]_3$ ternary Golay code
  or its dual. $M_{23}$ is the permutation group of the $[23,12, 7]$ binary Golay code and its dual. 

Since $|AGL(1,p)|=(p-1)p$, then $Per(C) \leq AGL(1,p)$ admits as order $pm$ with
$m|(p-1)$, hence it contains a Sylow subgroup $N$ of order $p$. From
Sylow's Theorem, $N$ is unique, hence it is
normal. For $p$ equals to 2 or 3, we have
$AGL(1,p)=S_p$, hence the codes are elementary. Then we assume $p \geq 5$. 
\end{proof}

If $(q,p)=1$, the number of cyclic codes of length $p$ over $\mathbb{F}_q$, 
is equal to $2^{\frac{p-1}{\ord_p(q)}+1}$, where
$\ord_p(q)$ is the multiplicative order of $q$ modulo $p$~\cite{guenda}. When
$\ord_p(q)=p-1$, there are only four codes, namely the elementary codes. This is the cases for $q=13$, $p=5,11,31,37,41$ or $q=11$ and $p=13,17,23,31$.

In the following table we give examples of permutation groups of cyclic codes of length $p$ over $\mathbb{F}_q$ 
 in the case $Per(C) \leq AGL(1,p)$. The integer $m$ is such that
 $Per(C) =C_m \ltimes C_p$. The codes are given in pairs, corresponding to the
 code and its dual.\\

\begin{tabular}{|c|c|c|c|}
\hline
 $q$&$p$& $m$& Codes \\
\hline
11&5  &2 & $[5,3,3],[5,2,4]$ \\
  & 5 &1  & $[5,2,4],[5,3,3]$ \\
  & 7 &3 & $[7,4,4],[7,3,5]$ \\
  &19  &3 & $[19,16,3],[19,3,6]$ \\
  & 37 &6 & $[37,31,5],[37,6,27]$ \\
\hline
   13 &7 &2 &$[7,5,3],[7,2,6]$;\\

&&& $[7,3,5],[7,4,4]$ \\
\hline
   & 17 &4 &$[17,13,4],[17,4,12]$;\\
&&& $[17,12,4],[17,5,11]$;\\
&&&$[17,8,8],[17,9,7]$\\
 & 17 & 8& $[17,9,8],[17,8,9]$\\
& 23 &11 & $[23,12,9],[23,11,10]$ \\
&29&14 & $[29,15,11],[29,14,12]$ \\
\hline
\end{tabular}\\

\begin{theorem}
\label{gr:p^r:primitif}
Let $C$ be a non elementary cyclic code over $\mathbb{F}_q$ of length
$n=p^r,$ where $r \geq 2$ and $q=p'^s$. Then $Per(C)$ is either:
\begin{enumerate}
\item an imprimitive group, and admits a system of $p^i$ blocks
each of length $p^s$, ($si=r$) formed by the orbits of\\ $<T^{p^i}>$, a such system of
  block is complete, or
\item a doubly transitive group equals to
 $$P\Gamma L(d,t),
\text{ with }p^r =\frac{t^d-1}{t-1}, d\geq 3, t
 \text{ a power of } p'.$$
 \end{enumerate}
\end{theorem}
\begin{proof}
A simply transitive subgroup of $S_n$ of degree $p^r$ which contains a
permutation of order $p^r$
(in our case $T$) must be imprimitive~\cite[p. 229]{dixon}. Hence it admits a system of  $p^i$ blocks each of
length $p^s$, ($si=r$) formed by the orbits of $T^{p^i}$~\cite{dobson}
p. 67, a such system of 
  block is complete. 
   
From the previous we have that if  $Per(C)$ is primitive it must be  doubly
   transitive. From Lemma~\ref{lem:sorley}, Lemma~\ref{lem:proj} and Remark~\ref{rem:iso} we get the results.
\end{proof}
\begin{corollary}
Let $p \geq 5$ be a prime number, $C$ a cyclic codes of length
$p^r,r \geq 2$ over $\mathbb{F}_p$. Then if $C$ is not an elementary code,
the group $Per(C)$ is imprimitive.
\end{corollary}
\begin{proof}
From Theorem~\ref{gr:p^r:primitif} if $Per(C)$ is primitive, hence
it is doubly transitive.  From Lemma~\ref{lem:sorley} the only cases when the
socle can be abelian are $N=C_p$ and $C_2 \times C_2$. In which cases
$Per(C)$ must be equal to $AGL(1,p)$ or $S_4$; which is impossible. From~\cite[Lemma 22]{dobson}, if $Per(C)$ is
doubly transitive with non abelian socle, then
$Soc(Per(C))=Alt(p^m)$, hence from Remark~\ref{rem:iso} the code is elementary.
\end{proof}

\begin{theorem}
\label{th:main2}
A non elementary cyclic code $C$ of composed length $n$ over $\mathbb{F}_q$, where $q=p'^s$,
  admits a permutation group $Per(C)$ which is either
\begin{enumerate}
\item imprimitive,
\item or doubly transitive; equals to   
$$Per(C)= P\Gamma L(d,t),$$
 with $n =\frac{t^d-1}{t-1}$, $d \geq 3$, $t$ a power of $p'$.
\end{enumerate}
\end{theorem}
\begin{proof}
The group $Per(C)$ contains a full cycle and is with composed
degree. Hence from Theorem of Burnside and Shur~\cite[ p. 65]{wielandt},
$Per(C)$ is either imprimitive or doubly transitive. In the doubly 
transitive from Lemma~\ref{lem:sorley} and Lemma~\ref{lem:proj} this happens only if
$Per(C)$ is $Alt(n)$, $S_n$ or $P\Gamma L(d,t).$ 
 \end{proof}

From~\cite{rolf} the permutation group of the Hamming code $[15,12,3]_2$
is $P\Gamma L(4,2)$.
 
Now we will deal with the permutation group of $MDS$ codes.
 
T. P. Berger~\cite{berger93} proved that the only linear $MDS$ codes with
minimum distance $d_C$ equals to $2$ or $n$ are the repetitions codes and their duals. Hence in the following result we assume that $n>d_C>2$. 
\begin{theorem}
Let $C$ be a non elementary  $[n,n-d_C+1,d_C]$ $MDS$ cyclic code over $\mathbb{F}_q$, $q=p'^s$,  $n>5$, $n>d_C>2,$ and
$gcd(n-2,d_C-2)=1$. If the group $Per(C)$
is doubly transitive, then 
\begin{enumerate}
\item $Per(C)=AGL(1,p)$ with $n=p$,
or
\item $Per(C)=PSL(d,p')$ with 
$n=\frac{p'^{d}-1}{p'-1}$, $d \geq 3,$ $gcd(p'-1,d)=1$.
\end{enumerate} 
\end{theorem}
\begin{proof}
If $n=p$, hence from Lemma~\ref{lem:sorley} $Per(C)$ is equal to
$AGL(1,p)$, $PSL(2,11)$, $M_{11}$, $M_{23}$ or a subgroup of $P\Gamma
L (d,p')$. From Proposition~\ref{lem:bur}, we eliminate $PSL(2,11)$,
$M_{11}$, $M_{23}$ since the $[11,6,5]_3$ and $[23,12, 7]_2$ Golay codes
are not $MDS$ and  $M_{11}$ does not correspond to any cyclic
code.
Now we need the following Lemma.
\begin{lemma}(\cite{berger93}) 
\label{Th:berger193}
If $C$ is an $MDS$ linear code with length $n$ and minimum
distance $d_C>2$ such that $Per(C)$ contains a doubly transitive
subgroup $B$
and $gcd(n-2,d_C-2)=1$,
then $B=Per(C)$. 
\end{lemma} 
 If $Per(C)$ is a doubly transitive not equal to $AGL(1,p)$,  then Theorem~\ref{gr:p^r:primitif},
 Theorem~\ref{th:main2}, Lemma~\ref{Th:berger193} and
 Remark~\ref{rem:iso} implies that the only solution is $Per(C)=PSL(d,p')=PGL(d,p')=P\Gamma L (d,p')$ for $gcd(p'-1,d)=1.$
\end{proof}
\begin{remark}
When $p$ a prime number any cyclic code of length $p$ over
$\mathbb{F}_p$ is $MDS$ and is equivalent to an extended Reed--Solomon
code~\cite{roth}. The permutation group of the last codes is the
affine group $AGL(1,p)$~\cite{berger91}.   
\end{remark}
\section{The Equivalence of Cyclic Codes }
Let $C$ and $C'$ two cyclic codes of length $p^r$, with $p$ an odd prime number.
The aim of this section is to find the permutations by which $C$ and
$C'$ can be equivalent. Even if our results are also true for the
cyclic combinatorial objects of length $p^r$, we consider only  cyclic
codes.

The multiplier $M_a$ is the affine transformation $M_a=\tau_{a,0}$. 
 It is obvious that the image of a
cyclic code by a multiplier is an equivalent cyclic code. If $gcd(n,
 \phi(n))=1$ or $n=4$, or $n=p.r$, $ p>r$ are primes, and the Sylow $p$-subgroup of $Per(C)$ has order $p$.
Then two cyclic codes $C$ and $C'$ of length $n$ are equivalent if
and only if they are equivalent by a multiplier~\cite{alspach,vanessa,palfy}.

When $C$ is a cyclic code of length $p^r$, $P$ a Sylow $p$-subgroup of $Per(C)$, the following subset of $S_{p^r}$ was introduced by N. Brand~\cite{brand} $$H(P)=\{ \sigma \in S_{p^r} |  \sigma ^{-1} T \sigma \in P \}.$$
The set $H(P)$ is well defined, since $<T>$ is
a subgroup of $Per(C)$ of order $p^r$, hence it is a $p$-group
of $Per(C)$. From Sylow's Theorem, there exists a Sylow $p$-subgroup
$P$ of $Per(C)$ such that $<T> \leq P$. Furthermore in some cases the
set $H(P)$ is a group.  
\begin{lemma}(\cite[Lemma 3.1]{brand})
\label{lem:brand}
Let $C$ and $C'$ two cyclic codes of length $p^r$. Let
$P$ be a Sylow
$p$-subgroup of $Per(C)$ which contains $T$. Then $C$ and $C'$ are
equivalent if and only if $C$ and $C'$ are equivalents by an element
of $H(P)$. 
\end{lemma}

In the case of length $p^2$, $H(P)$ has been given explicitly by
 Huffman et al.~\cite{vanessa}. They proved that if $P$ is a  Sylow $p$-subgroup of $Per(C)$ and $|P|=p^2$, then the cyclic codes $C$ and $C'$ are equivalent if 
 only if they are equivalent by a multiplier. If $p^2<|P| \leq
 p^{p+1}$ the codes $C$ and $C'$ are equivalent by
 a power of a multiplier times a power of a generalized
 multiplier. A generalized multiplier is a permutation 
  $\mu_a^{(p^k)} \in S_{p^r}$ defined  as follows:
 \begin{equation}
\begin{split}
          i+bp^k &\longmapsto (ai) \mod p^k + b p^k.
\end{split}
\end{equation} 
 $a \in Z_{p^k}^*$, $j=i +bp^k \text{ with } k \leq r, 0 \leq  i<p^k \text{ and } 0 \leq
  b<p^{r-k}$. Any integer $0
  \leq j < p^r-1$, can be uniquely decomposition as above. 

 Further for $c \in  Z_{p^k}^*$ the generalized affine transformation
 is $\mu_{a,c}^{(p^k)}$ is defined by :
 \begin{equation}
\begin{split}
          i+bp^k &\longmapsto (ai+c) \mod p^k + b p^k .
\end{split}
\end{equation}

Now, we are interested on the existence of the multiplier $M_{p+1}$ in
the group $Per(C)$. Because it has as
effect to simplify Algorithm 3.1 for the cyclic codes and
Algorithm 6.1l for extended cyclic objects~\cite{vanessa}.
\begin{lemma}
\label{lem:ord}
Let $z$ be the largest integer such that $p^z | (q^t-1)$, with
$t$ order of $q$ modulo $p$. Hence if $z=1$ we have :
\begin{enumerate} 
\item $\ord_{p^r}q=p^{r-1}t$.
\item The multiplier $M_{p+1} \in Per(C)$. 
\end{enumerate}
\end{lemma}
\begin{proof}
 The proof of $z=1 \Rightarrow \ord_{p^r}q=p^{r-1}t$, follows
 from~\cite[Lemma 3.5.4]{guenda}.

 Since $z=1 \Rightarrow \ord_{p^r}q=p^{r-1}t$, we choose integers $a_k,$ for $ k=0,
\ldots, p^r -1$, such that
$q^{kt}-1 =a_kp$. Since $\ord_{p^r}q=p^{r-1}t$, the integers $q^{kt}-1$, $k=0, \ldots,
p^{r-1}-1$ are all different modulo $p^r$. This gives the integers $a_k$,
$k=0, \ldots, p^{r-1}-1$ are different modulo $p^{r-1}$. Hence there
exists $a_{k_0}$ such that, $a_{k_0}\equiv 1 \mod  p^{r-1}$, then
$a_{k_0}p=p \mod p^r \iff q^{k_0t}-1 \equiv p \mod p^r$. So $
q^{k_0t}=(1+p) \mod p^r$.  Now by multiplying by $ip^j$ we obtain
:
$$(1+p)q^j i=iq^{j+k_0t} \mod p^r,$$ i.e., the cyclotomic class of $i$
is invariant by the
multiplier $M_{p+1}$. 
\end{proof}

The following proposition gives some properties of the set $H(P)$.
\begin{proposition}
\label{prop:norm}
Let $P$ be a Sylow subgroup of $Per(C)$ which contains $T$. Hence the group
$H(P)$ verifies :
\begin{enumerate}

\item $AG(p^r)=N_{S_{p^r}}(<T>) \subset H(P)$,
\item $N_{S_{p^r}}(P) \subset H(P)$.
\end{enumerate}
\end{proposition}
\begin{proof}
 The fact $AG(p^r)=N_{S_{p^r}}(<T>)$ follows from~\cite[p. 710]{huffman}.

Now we consider $\sigma \in N_{S_{p^r}}(<T>)$, then $$\sigma ^{-1} <T>
\sigma = <T>.$$ Since we assumed $<T>\leq P$ , hence $\sigma \in H(P)$.

Let $ \sigma \in N_{S_{p^r}}(P)$, hence $\sigma \in N_{S_{p^r}}(P)$ verifies
$\sigma ^{-1} P \sigma = P$. Since we assumed $<T>\leq P$, then
$\sigma ^{-1} T \sigma \in P$, i.e.
\begin{equation}
 N_{S_{p^r}}(P) \subset H(P).
\end{equation}
\end{proof}

\begin{lemma}
\label{lem:groupe}(\cite[Theorem.5.6]{vanessa})
Let $p$ be an odd prime  and let $q$ be a prime power with $p \nmid
q$. Let $C$ be a cyclic code of length $p^r$ over $\mathbb{F}_q$ and $t_k$ be the order of $q$ modulo
$p^k$ and suppose that $z=1$. For $1\leq k \leq r$, $Per(C)$
contains the group $G_k=\{ \mu_{q^i,c}^{(p^k)} \,|\, 0 \leq i <t_k, 0
\leq c <p^k\}$ which is of order $t_kp^k$. Furthermore, each element of
$H_k =\{  \mu_{q^i,0}^{(p^k)} \,|\, 0 \leq i <t_k \}$ fixes the
idempotent of the code $C$.    
\end{lemma}
\begin{theorem}
\label{th:gr.sy}
Let $C$ be a cyclic  code of length $p^r, r >1$. Hence a Sylow $p$-sous
group of $Per(C)$ is of order $p^s$, such that 
\begin{equation}
r \leq s \leq
p^{r-1}+p^{r-2}+ \ldots  +1.
\end{equation}
Furthermore if $z=1$, $s$ verifies:
$$2r-1 \leq s \leq p^{r-1}+ p^{r-2}+ \ldots 1.$$    
\end{theorem}
\begin{proof}
Let $P$ be a Sylow $p$-subgroup of $Per(C)$, hence $P$ is a $p$-group of $S_{p^r}$.
From Sylow's Theorem  $P$ is contained in a Sylow $p$-subgroup of
$S_{p^r}$. It is well known that a Sylow $p$-subgroup of $S_{p^r}$ is of order
$p^{p^{r-1}+p^{r-2}+ \ldots  +1}$. It is isomorphic with the wreath
product $\mathbb{Z}_p \wr \ldots \wr \mathbb{Z}_p$~\cite[Kalu\v{z}nin's Theorem]{robinson}, hence $r \leq s
\leq p^{r-1}+p^{r-2}+ \ldots  +1$.

If $z=1$, from Lemma~\ref{lem:groupe} we have that $Per(C)$ contains
the group $G_r$, of order $t_rp^{r}$. From Lemma~\ref{lem:ord} we
have that $t_r=t_1p^{r-1}$, then we have  $|G_r|=t_1p^{2r-1}$. Hence 
 $p^{2r-1}$ divides $|Per(C)|$, which means that $Per(C)$ contains a
$p$-subgroup of order at least $p^{2r-1}$.  
\end{proof}

For $n<p$, we define the following subsets of $S_{p^r}$:\\
$Q^n=\{f:Z_{p^r} \rightarrow Z_{p^r} | f(x)=\sum_{i=0}^{n}a_ix^i,\\
a_i \in Z_{p^r}\text{ for each }i, p  \text{ is relatively prime to }
a_1, \text{ and }\\p^{r-1}
\text{ divides }a_i\text{ for }i=2,3,\ldots, n\}.$
$$Q_1^n=\{f \in Q^n | f(x)=\sum_{i=0}^{n}a_ix^i, \text{ with }a_1 \equiv 1
\mod p^{r-1} \}.$$
The set $Q^n$ and $Q_1^n$ are subgroups of $S_{p^r}$~\cite[Lemma 2.1]{brand}.
\begin{theorem}
\label{lem:Qgroup}
Let $Per(C)$ be the permutation group of a cyclic code of length
$p^r$ and $P$ a Sylow $p$-subgroup of $Per(C)$ of order $p^s$ such that
$T \in P$. If $r\leq s\leq p+1$, then $H(P)$ is a group and we have:

$(a)$ if $s=r, P=<T>$, and $H(P)=AG(p^r)$,

$(b)$ if $s >r$, $P=Q_1^{s-2}$, and $H(Q_1^{s-2})=Q^{s-1}$.
\end{theorem}
\begin{proof}
From Theorem~\ref{th:gr.sy} we have that $r \leq s \leq p^{r-1}+ p^{r-2}+
\ldots 1.$ In the case $r=s$, it is obvious to have $P=<T>$, and
$H(<T>)=N_{S_n}(<T>)$. Proposition~\ref{prop:norm} gives $H(<T>)=AG(p^r)$. Furthermore, we remark that the part $(b)$
of~\cite[Theorem 2.1]{vanessa} can be generalized
for the case $p^r$. Since it is based essentially on the \cite[Lemma. 2.4]{vanessa} and the~\cite[Lemma. 3.2]{brand}( the last Lemma was
given for the length $p^r$). The hypothesis $s\leq
p+1$ is to assure that $Q_1^{s-1} \leq Per(C)$. Hence the result
follows from~\cite[Theorem 2.2]{brand}.    
\end{proof} 

As noticed Brillhart et al.~\cite{brillhart}, it is quite unusual to
have $z >1$. Hence the importance of the following result. 
\begin{theorem}
Let $C$ be a cyclic code of length $n=p^r$, over $\mathbb{F}_q$, such
that $gcd(p,q)=1,z=1$ and $P$ be a Sylow
$p$-subgroup of $Per(C)$ of order $p^{2r-1}$, then we have :
\begin{enumerate}
\item $P$ is normal in $G_r$,
\item $N_{S_n}(P)=G_r$.
\item $H(P)=\{\tau \in S_n | \tau:i \mapsto q^{ij}
 a+q^{(i-1)j}c+q^{(i-2)j}c+\ldots+c \, , 0\leq j<tp^{r-1},\, a,c\in Z_n \}$
\end{enumerate}
\end{theorem}
\begin{proof}
In this case we can assume that $P \leq G_r$, because we have $T= \mu_{1,1}^{(p^r)}\in
G_r$. Let $N$ be the number of the Sylow 
$p$-subgroup in $G_r$. From the Sylow's Theorem,
$N\equiv 1 \mod p$ and $N$ divide $ |G_r |=t_1p^{2r-1}$. Assume that $N=1+kp,$
with $ k > 0$ and
$N\alpha=t_1p^{2r-1}, \alpha \geq 1$. Hence $(1+kp)\alpha=
t_1p^{2r-1}$, but $((1+kp),p^{2r-1})=1$. Then $(1+kp) \alpha'=t_1$,
absurd. Because $t_1 | (p-1)$ and $\alpha' \geq 1$, hence
$N=1$. Since $N=1$ and the Sylow subgroups are conjugate this gives
that $P$ is normal in $G_r$. 

Now, we will prove that $N_{S_n}(P)=G_r$ :

Let $\sigma \in P \leq G_r$, then we can write $\sigma =
\mu_{q^j,c}^{(p^r)}=T^c  M_{q^j}$, with $T$ the shift and
$M_{q^j}$ is a multiplier.  Let
$\tau \in N_{S_n}(P)$, hence there exists $\sigma' = T^{c'}
M_{q^{j'}} \in P$
such that  $\tau \sigma = \sigma ' \tau$,  this is equivalent to $$\tau  T^{c} 
 M_{q^j}= T^{c'}  M_{q^{j'}} \tau. $$
For $0$ we have :
 $$\tau  T^{c} M_{q^j} (0)= \tau  T^{c}(0)=\tau(c),$$ 
hence we have:

$$\tau (c)= T^{c'} M_{q^{j'}}\tau(0)= q^{j'}\tau (0)+c'.$$

Then, $\tau (c)= cq^{j'} +d$, where $d=c'+q^{j'}(\tau (0)-c)$. This implies $N_{S_n}(P)
\leq G_r$. 

 Furthermore $P$ is normal in $G_r$, which gives that $
 N_{G_r}(P)=G_r$. But
$N_{G_r}(P) \leq N_{Per(C)}(P)  \leq  N_{S_n}(P)\leq G_r$, then
$$N_{S_n}(P)=G_r.$$
From the Lemma~\ref{prop:norm} we have that $N_{S_n}(P) \subset
H(P)$. Hence $N_{S_n}(P)=G_r \subset H(P)$.

 Now, let $\tau \in H(P)$, hence there exist $c$ and $j$ such that :
 $$\tau T \tau^{-1} =  T^c   M_{q^j} \iff  \tau T = T^c
 M_{q^j}\tau.$$That is because $P\leq G_r$.
For $0$ we obtains :\\
$\tau T(0)=\tau(1)=q^j \tau(0)+c$.\\
$\tau T(1)=\tau (2)=T^c M_{q^j} (\tau(1))=q^j\tau(1)+c$\\
$\tau T(2)=\tau (3)=T^c M_{q^j}(\tau(2))=q^j\tau(2)+c $\\
$\vdots$\\
$\tau (i)=q^{ij} (\tau(0))+q^{(i-1)j}c+q^{(i-2)j}c+\ldots+c$.\\
Then the elements of $H(P)$ are $\tau \in S_n$ such that
  $\tau (i)=q^{ij} a+q^{(i-1)j}c+q^{(i-2)j}c+\ldots+c.$
\end{proof}
\section{The Equivalence for the Quasi-Cyclic Codes}
A code $C$ of length $n=lm$ is said to be quasi-cyclic of order $l$ over
$\mathbb{F}_q$, if it is invariant by the
permutation
 \begin{equation}
\begin{split}
T^l : Z_n & \longrightarrow  Z_n \\
          i &\longmapsto i +l \mod n.
\end{split}
\end{equation}
We consider the cycles $\sigma _i=(i, i+l,i+2l \ldots , i+
(m-1)l)$ for $0 \leq i \leq l-1$. The cycles $\sigma _i$ have order $m$. Furthermore we have
\begin{equation}
\label{eq:tl} 
T^l= \sigma _0 \ldots \sigma_{l-1}.
\end{equation}
 This gives that 
$|(T^l)|=lcm(|\sigma _0|, \ldots,| \sigma_{l-1}|)=m$.

\begin{proposition}
\label{prop:aff.qua}
Let $n=lm$, with $(m,l)=1$ and  $<T^l>$ the subgroup of $S_n$
generated by the permutation $T^l$. Hence the normalizer of $<T^l>$ in $S_n$ contains the following groups.
\begin{enumerate}
 \item $Q=< \sigma_0, \ldots, \sigma_{l-1},T>.$
\item $AG(n)$.
\end{enumerate}
\end{proposition}
\begin{proof}
This is obvious that $T\in N_{S_n}(<T^l>)$. Now we consider $\sigma_i
\in Q,$ from the relation (\ref{eq:tl}) we have $\sigma_0
\ldots \sigma_{l-1}=T^l$. Furthermore the cycles $\sigma_i$ are
disjoints, then commute. Hence $\sigma_i^{-1}
T^l \sigma_{i}=T^l$.

We consider the affine transformation $\tau _{a,b} \in AG(n)$
proving that $\tau _{a,b} \in N_{S_n}(<T^l>)$ is equivalent to prove
the existence of an $\alpha \in \mathbb{N}^*$ such that, $$\tau
_{a,b}T^l\tau _{a,b}^{-1}=T^{l\alpha}.$$
The permutation $\tau _{a,b}$ can be 
decomposed  as follows,
$\tau _{a,b}=\tau _{1,b}\tau _{a,0}.$ 
Which gives with (\ref{eq:tl}) the following equality :
\begin{equation}
\label{eq:tau}
 \tau_{a,b}T^l\tau _{a,b}^{-1}=  \tau _{1,b}\tau _{a,0}\sigma _0 \ldots \sigma_{l-1}.
\tau _{a,0}^{-1}\tau _{1,b}^{-1}.
\end{equation}
This is well known~\cite[Lemma. 5.1]{hall} if $\sigma= \sigma _0\ldots
\sigma_{l-1}$ is a product of the disjoint cycles and $S$ is a
permutation of $S_n$. Hence
$S\sigma S^{-1}= S(\sigma_0)S(\sigma_1) \ldots S(\sigma_{l-1})$.
For $r_a= a \mod l$ we obtains that $\tau _{a,0}(\sigma_i)=\sigma
 ^a_{ir_a}$. This gives :
$$\tau _{a,0}(\sigma_0)\tau _{a,0}(\sigma_1) \ldots \tau
 _{a,0}(\sigma_{l-1})=\sigma_0^a\sigma_{ra}^a \ldots \sigma_{ra(l-1)}^a=T^{la}.
$$
For $r_b=b \mod l$, we obtain $$ \tau _{1,b}\sigma _i\tau
 _{1,b}^{-1}=\tau_{i+r_b},$$ hence $$\tau _{1,b} T^l
 \tau _{1,b}^{-1}=\prod_{i=0}^{l-1} \sigma _{i+r_b}=T^l.$$
Finally we obtains :
$$\tau_{a,b}T^l\tau _{a,b}^{-1}=\tau _{1,b}\tau _{a,0}T^l\tau _{a,0}^{-1}\tau _{1,b}^{-1}=\tau_{1,b}T^{la}\tau_{1,b}^{-1}=T^{la}. $$
This gives $\alpha =a,$ hence $\tau_{a,b} \in N_{S_n}(<T^l>).$ 
\end{proof}

In~\cite{chabot} C. Chabot gave explicitly the group $N_{S_n}(<T^l)>$.
\subsection{Quasi-Cyclic Codes of length $p^rl$}
In the following, we consider the quasi-cyclic codes of length
$n=p^rl$, with $p$ a prime number which verifies $(p,l)=(p,q)=1$.
In this case $<T^l> \leq Per(C)$ is a subgroup of order $p^r$.
Hence it is contained in a Sylow $p$-subgroup $P$.
\begin{lemma}
Let $C$ and $C'$ two quasi-cyclic codes of length $n=p^rl$ and
$P$ a Sylow $p$-subgroup of $Per(C)$ such that  $T^l \in P$. Hence $C$
and $C'$  are equivalent only if they are equivalent by the elements
of the set
$$H'(P)=\{\sigma \in S_n | \sigma ^{-1} T^l \sigma \in P \}.$$
\end{lemma}
\begin{proof}
Since $C$ and $C'$ are equivalent, hence there exist a permutation
$\sigma \in S_n$ such that $C'=\sigma(C)$. This gives the relation
between the permutation groups $Per(C)$ and $Per(C')$.
\begin{equation}
\label{eq:eq}
Per(C')=\sigma Per(C) \sigma ^{-1}
\end{equation}
Let $P$ be a Sylow subgroup of $Per(C)$, hence from the
relation~\ref{eq:eq} we have $\sigma P \sigma ^{-1}=P^{''}$ is a Sylow
$p$-subgroup of $Per(C')$. From the Sylow Theorem there exists $\tau
\in Per(C')$ such that $\tau P' \tau ^{-1}=P^{''}.$ We can assume that  $<T^l> \leq
P'$, since $<T^l>$ is a $p$-group.
Let $\gamma= \tau ^{-1} \sigma$, then $\gamma$ is an isomorphism between
$C$ and $C'$, because $\gamma(C)=\tau ^{-1} \sigma(C)=\tau ^{-1} C'=C'$
and $\gamma ^{-1} T^l \gamma=  \sigma ^{-1} \tau T^l \tau ^{-1} \sigma
\in \sigma ^{-1} P^{''} \sigma =P$ (because $\tau T^l \tau ^{-1} \in \tau
P' \tau ^{-1}$,) hence $\gamma \in H'(P).$  
\end{proof}

It is obvious that if $P=<T^l>$, then we have
$$N_{S_n}(<T^l>)=H'(<T^l>).$$

The following proposition gives other properties of $H'(P)$.
\begin{proposition}
\label{prop:norm2}
Let $P$ be a Sylow $p$-subgroup of $Per(C)$, hence the group
$H'(P)$ verifies the following proprieties:
\begin{itemize}
\item $N_{S_n}(<T^l>) \subset H'(P)$,
\item $N_{S_n}(P)\subset H'(P)$. 
\end{itemize}
\end{proposition}
\begin{proof}
We consider $N_{S_n}(<T^l>)$, the normalizer of $<T^l>$ in $S_n$. Then 
the permutation $\sigma
\in N_{S_n}(<T^l>)$ verifies $\sigma ^{-1} <T^l> \sigma = <T^l>
\subset P$. Hence,
\begin{equation}
 N_{S_n}(<T^l> ) \subset H'(P).
\end{equation}
We consider $N_{S_n}(P)$, the normalizer of $P$ in $S_n$. Then the permutation $\sigma
\in N_{S_n}(P)$ verifies $\sigma ^{-1} P \sigma = P$, hence for
$T^l\in P$ we have $\sigma ^{-1} T^l\sigma \in P$ . Hence
\begin{equation}
 N_{S_n}(P) \subset H'(P).
\end{equation}
\end{proof}

\begin{corollary} The set $H'(P)$ verifies
 $AG(n) \subset H'(P)$.
\end{corollary}
\begin{proof}
From Proposition~\ref{prop:aff.qua} we have $AG(n) \leq
N_{S_n}(<T^l>)$. Furthermore, from Proposition~\ref{prop:norm2} we
have $N_{S_n}(<T^l> ) \subset H'(P)$, hence the result.
\end{proof}

By using the software GAP, we find on several examples that the set
$H'(P)$ is an imprimitive group. By using this conjecture we prove the following result.   
\begin{proposition}
If $n$ is even then, $H'(P)$ is either
\begin{enumerate}
\item imprimitive,
\item or $S_n$ if the code is trivial.
\end{enumerate}
If  $n$ is odd then, $H'(P)$ is either
\begin{enumerate}
\item imprimitive,
\item  or $H'(P)=Alt(n)$,
 \item or $S_n$ if the code is trivial.
\end{enumerate}
\end{proposition}
\begin{proof}
From Proposition~\ref{prop:aff.qua} $H'(P)$ contains the group
$$Q=< \sigma_0, \ldots,\sigma_{l-1},T>, $$ by a similar proof to the
 Theorem~\ref{th:main2} we obtains that the group $H'(P)$ is either
 imprimitive or doubly transitive.
In the case doubly transitive we need the following Lemma.
\begin{lemma}(\cite{williamson})
\label{lem:williamson}
Let $G$ be a primitive group of degree $n$ which contains a cycle of
length $m>1$. Hence if $m$ verifies $m<(n-m)!$, we have $G=Alt(n)$ or $S_n$.   
\end{lemma}
If the group $H'(P)$ is doubly transitive, then it is
primitive. From Proposition~\ref{prop:norm2} it
contains the cycles $\sigma_i$ of length $p^r$. Furthermore we have
$p^r<(n-p^r)!$. Hence from Lemma~\ref{lem:williamson} $H'(P)$ is
either $Alt(n)$ or $S_n$. If $n$ is even, since $T\in H'(P)$ is a
cycle of length $n$ is odd. Then $T \notin Alt(n)$, hence
$H'(P)=S_n$. If $n$ is odd, $H'(P)$ is the group $Alt(n)$ or $S_n$.
\end{proof} 
\section{Conclusion}
The aim of this work is to find solutions of the three following  problems:
\begin{itemize}
\item[$\bullet$]\textbf{Problem I} The classification of the permutation groups of
cyclic codes.
\item[$\bullet$]\textbf{Problem II} The determination of the permutations by which two cyclic codes of length $p^r$ can be equivalent.
\item[$\bullet$]\textbf{Problem III} The determination of the permutations by which two quasi-cyclic codes can be equivalent.
\end{itemize}
Our contribution to the solution of Problem allows to solve all the 
primitive cases and some imprimitive cases. Even though, there
is still work and extension to do for the imprimitive cases.
 
As solution to Problem II we found
explicitly the set $H(P)$ of permutations by which two cyclic codes of
length $p^r$ can be equivalent. We also proved that is
sometimes a group. 

 We think that our contribution to solve Problem III can be refined by proving that the set $H'(P)$ of permutations by which two quasi-cylic codes are equivalent is a group. If it is the
case, since we proved that it will be an imprimitive group or $Alt(n)$ or
$S_n$, hence it will be interesting to find some of other properties
of $H'(P)$ in the imprimitve cases. Also, the results can be
generalised to other situations of length.     
\section*{Acknowledgment}
The author would like to thank Professor Thierry P. Berger for help
and encouragement in carrying out the research in this paper.
{\small

\begin{IEEEbiography}{Kenza Guenda}

\end{IEEEbiography}
\end{document}